\documentclass[sn-mathphys,Numbered]{sn-jnl}

\usepackage{graphicx}%
\usepackage{multirow}%
\usepackage{amsmath,amssymb,amsfonts}%
\usepackage{amsthm}%
\usepackage{mathrsfs}%
\usepackage[title]{appendix}%
\usepackage{xcolor}%
\usepackage{textcomp}%
\usepackage{manyfoot}%
\usepackage{booktabs}%
\usepackage{algpseudocode}%
\usepackage{listings}%
\usepackage{caption} 
\usepackage{subcaption}
\usepackage{amsmath,amsfonts}
\usepackage{amsthm,amssymb}

\newcommand{\E}{{\mathbb E}}
 
\theoremstyle{thmstyleone}%
\newtheorem{theorem}{Theorem}
\newtheorem{proposition}[theorem]{Proposition}%

\theoremstyle{thmstyletwo}%
\newtheorem{remark}{Remark}%

\theoremstyle{thmstylethree}%
\newtheorem{definition}{Definition}%

\begin{document}

\title[ ]{Distributional Reinforcement Learning on Path-dependent Options}


\author*[1 ]{\fnm{Ahmet Umur} \sur{\"Ozsoy}}\email{umurozsoy@gmail.com}


\affil*[1]{\orgdiv{Department of Industrial Engineering}, \orgname{Gebze Technical University}, \orgaddress{  \city{Kocaeli}, \postcode{41400},  \country{Turkey}}}



\abstract{ We reinterpret and propose a   framework for pricing path-dependent financial derivatives by estimating the full distribution of payoffs using Distributional Reinforcement Learning (DistRL). Unlike traditional methods that focus on expected option value, our approach models the entire conditional distribution of payoffs, allowing for risk-aware pricing, tail-risk estimation, and enhanced uncertainty quantification. We demonstrate the efficacy of this method on  Asian  options, using quantile-based value function approximators.}





\maketitle

\section{Introduction}\label{sec1}

Distributional\footnote{This is a preliminary draft of an ongoing study. Comments and suggestions are welcome; a revised and expanded version will follow.} Reinforcement Learning (DistRL) stands a unique and an interesting extension of Reinforcement Learning (RL).
Traditionally, RL algorithms aims to learn the expected value of returns, rather than the distribution of returns \footnote{For the time being, we do not refer to financial returns, but the returns in RL modeling.}.
However, an interesting proposal on an extension of  reinforcement learning appeared in \cite{bellemare2017distributional} along with the concept of distributional Bellman equation, an adjusted\footnote{With actions omitted.} version for our setting is given in Equation \eqref{distBellman_intro}.
Although distributional perspectives on RL have a long-standing history, \cite{bellemare2017distributional} stands the first disciplined introduction of distributional reinforcement learning.
For a historical account, we refer to the references in \cite{bellemare2017distributional}.
The study offers a novel method to model the value distribution\footnote{The distribution of interest is coined \emph{value distribution} by \cite{bellemare2017distributional}, and we follow the terminology in this study.} via categorical approximation, and proves the distributional Bellman operator is a contraction in the Wasserstein metric.
A study that closely came after, \cite{dabney2017distributional}, offers a quantile function approximation (Quantile Regression - Deep Q-Network, QR-DQN) rather than a categorical given in \cite{bellemare2017distributional}.
The suggestion naturally minimizes the Wasserstein distance, along with more stable learning.
This especially becomes valuable in financial context, given the nature of the underlying data.
However, we remark that neither of these studies utilizes financial data.

Although we are closely aligned with \cite{dabney2017distributional}, we differ in our architectural choice. 
In contrast to prior works such as \cite{dabney2017distributional}, which used deep neural networks to approximate quantile functions, we introduce a Radial Basis Function (RBF) feature expansion for quantile learning.
This modeling choice provides an interpretable and a computationally efficient architecture, and is particularly well-suited to financial applications where domain knowledge can guide the design of the feature space.
Perhaps most importantly, it is lighter and more stable than training a large deep network and perfectly aligned with option pricing where states have interpretable meaning. 
To our knowledge, this is the first application of RBF-based quantile approximation in Distributional RL for option pricing.
Speaking of option pricing applications via  RL are numerous, for instance \cite{marzban2023deep}, \cite{du2020deep}, \cite{Cao_2020} and \cite{ozsoy2025qlbs} among others.
However, at the time of writing, we are only aware of the study \cite{cao2023gamma} on (deep) distributional perspective of RL.
Through quantile regression,  \cite{cao2023gamma} tackles Gamma and Vega hedging, a framework that is completely different than ours and making our study stand unique in the literature.

Aligning the principled envision of DistRL brings us to our inherent purpose.
Path-dependent options, a special class of exotic options, derive their payoff from the trajectory of the underlying asset rather than its terminal value.
Although this distinct feature offers an interesting alternative in the financial markets, such options are packed with further complexities.
Therefore we remark that these options pose a significant challenge in financial engineering due to the high dimensionality and complexity of their payoff structures.
Traditional numerical approaches such as Monte Carlo simulations or PDE-based methods ,  aim to estimate the expected (discounted) payoff under the risk-neutral measure, $\mathbb{Q}$, for instance \cite{alziary1997pde}, \cite{rogers1995value} and \cite{vecer2001new}.
Although, in terms of pricing this could suffice; estimating such expected values while ignoring the distributional perspectives embedded neglects the valuable information.
This issue could further escalate in cases where tail-risk or tail-related properties are prioritized.

To address this, we leverage tools from Distributional Reinforcement Learning (DistRL) to estimate the entire (conditional) distribution of the payoff of a path-dependent option. 
With this, we frame the problem of option pricing as a distributional value estimation problem in a Markovian state space.
We formulate  the evolution of the terminal payoff distribution as a function of the current state and path-dependent features and treat it as a learning problem.
In Distributional RL, we consider the return distribution, with some notations explained in detail in the next section,
%
\begin{equation}
Z^{\pi}(s)=\left(\sum_{t=0}^{\infty} \gamma^{t} R_{t} \mid s \right)
\end{equation}
And define the distributional Bellman operator:
\begin{equation}\label{distBellman_intro}
\mathcal{T}^{\pi} Z(s) \stackrel{D}{=} R(s)+\gamma Z(s^{\prime})
\end{equation}
This operator is a contraction in the Wasserstein metric $W_{p}$ (Bellemare et al., 2017), and forms the backbone for learning quantiles or categorical approximations of $Z$.
Although, we develop the approach in the next section; so as not to cause confusion we remark that in our set-up we omit the policy associated, $\pi$, as there is no agent choosing actions.
That is, In our framework, the system evolves under a stochastic process (e.g., geometric Brownian motion), and no control policy is applied.
The value distribution $Z(s)$ reflects the uncertainty induced by stochastic dynamics and path-dependent payoffs. 
Relevance to option pricing, then, reduces to estimating $$Z(s_{0})=\mathbb{P}(f(S_{0:T}) \leq z \mid S_{0}=s_{0})$$ for various levels of $z$. 
Namely, to learn the cumulative distribution of future payoff conditioned on current state.
This goes beyond just estimating $\mathbb{E}[f(S_{0:T})]$   as it estimates (the entire) distributional functional.


From another perspective, our manuscript provides and elaborates a rigorous connection between distributional reinforcement learning (DistRL) and conditional distributions of path functionals.
It does it so in terms of financial context for the first time to the best of our knowledge.
We, in our humble opinions, believe it as the missing link between model-free approaches and measure-theoretic models in finance.
Given the black-box type of modeling ongoing lately, we, while building the model, provide much further reasoning   aligned with the requirements of financial modeling; and in doing so we aim to lay ground for later studies.   
Our study elaborates centers on learning the full conditional distribution of a payoff functional $( f(S_{0:T}) )$, given a filtration-reduced, finite-dimensional state $( s_t \in \mathbb{R}^d )$ that embeds relevant path information (e.g., running average, maximum, or time index).

Equation \eqref{distBellman_intro}, in a simple but effective way, maps the distributional Bellman operator into a recursive form that operates over probability measures induced by stochastic processes in our analogy.
Simply, \eqref{distBellman_intro} mimics the law of iterated expectations; yet on distributional level. 
Further, the recursion could also be restated as a stochastic flow of (conditional) measures; in continuous-time this becomes the Fokker-Planck equation for the evolving density of payoffs obtained at maturity.
Therefore, on grounds of stochastic analysis; the recursion is sound and highly applicable.
In financial terms, for instance; this allows path-dependent option prices to be modeled as evolving laws of random variables, offering direct access to quantiles, tail risk, and risk-adjusted measures such as CVaR, which are essential in modern risk management and regulatory frameworks.
Mathematically, the value distribution process $ \{Z_t\}_{t \in [0,T]} $ can be interpreted as a measure-valued stochastic process $Z_t := \mathcal{L}^\mathbb{Q}(f(S_{0:T}) \mid s_t)$ which evolves randomly due to the underlying dynamics of $ S_t $. 
However, our approach tracks this process via a family of quantile functions$ \theta_i(s_t)$, producing a tractable low-rank approximation to a family of conditional probability laws on $ \mathbb{R} $.

Therefore, we \emph{reinterpret} the recursive pricing of path-dependent options, specifically Asian options, as a distributional fixed-point problem on the space of probability measures. 
With this, on narrower level, we aim to present and reframe a pricing approach for a challenging option along with laying the fundamental theory for future finance-related studies.
On higher level we aim to achieve quantifying full distributional uncertainty rather than point estimates in financial context, encoding domain knowledge through feature-engineered state summaries $ s_t $  and adapting distributional estimates  without requiring parametric density modeling.
This contribution via reframing of the stochastic  problem as recursive quantile learning model constitutes a novel intersection between stochastic process, computational finance, and learning-based control.

The ability to learn and represent conditional payoff distributions $ Z_t := \mathcal{L}(f(S_{0:T}) \mid s_t) $ carries several practical benefits across derivative pricing, hedging, and model calibration.
Once trained, the model enables fast inference of distributional payoffs at any state $ s_t $. 
This actually permits real-time pricing of path-dependent derivatives during market events, which is typically infeasible with nested simulation or full revaluation methods.
The finite-dimensional summary $ s_t \in \mathbb{R}^d $ can be augmented to capture domain-specific path features  allowing our framework to accommodate exotic or path-sensitive options with minimal architectural changes.

Another possibility is gradient-based calibration as the quantile functions $ \theta_i(s)$ are smooth (via RBF parametrization in our setting), one can use gradient-based optimization to calibrate model parameters (e.g., implied vol surface or jump parameters) by matching market-implied quantiles or distributional moments, without relying on expectation-based loss.
In financial markets with significant discontinuities (i.e. jupms), the stochastic updates induced by the quantile regression loss  is likely to remain well-behaved. 
Extreme events appear as fat-tailed shocks in the state process  shifting the learned quantile estimates without requiring a re-specification of the parametric state transition. 
Hypothetically, the DistRL approach provides robust adaptability to jumps, leveraging its distributional representation.
However, these remarks of ours on both jumps and advantages of calibration require further research, and we leave such matters aside as a further study.

\subsection{Overview of the Learning Framework}
We now summarize the core components of our Distributional RL framework for pricing path-dependent derivatives:
\begin{itemize}
	\item \textbf{State Representation:} The process is made Markovian via augmentation (e.g., $s_t = (S_t, A_t, t)$).
	\item \textbf{Value Target:} The goal is to learn $Z_t(s_t) = \mathcal{L}(f(S_{0:T}) \mid s_t)$.
	\item \textbf{Distributional Approximation:} $Z(s) \approx \sum_i \delta_{\theta_i(s)}$, with $\theta_i(s) \approx F_Z^{-1}(\tau_i)$.
	\item \textbf{Training Objective:} Minimize quantile regression loss $\mathcal{L}_{\mathrm{QR}}$ using SGD.
	\item \textbf{Function Approximation:} $\theta_i(s) = w_i^\top \phi(s)$, with RBF features $\phi(s)$.
\end{itemize}
This structure allows us to rigorously approximate the full return distribution for  payoffs using fixed quantile targets and semi-gradient updates.
The full conditional distribution over payoffs allows for robust risk assessment and tail-aware pricing strategies. In what follows, we develop a quantile-based recursive algorithm for approximating, tailored to pricing Asian options, where the Markovian state includes both current spot price and running average.

\section{Formulating the Problem}

\subsection{Preliminary Perspectives}

Let $(\Omega, \mathcal{F}, \{\mathcal{F}_t\}_{t\in[0,T]}, \mathbb{Q})$ be a filtered probability space satisfying the usual conditions, where $\mathbb{Q}$ is the risk-neutral measure.
We consider a financial asset whose (risk-neutral) price process $\{S_t\}_{t\in[0,T]}$ follows the stochastic differential equation, \cite{shreve2004stochastic}, below such that
\begin{equation}\label{eq:sde}
dS_t = rS_t dt + \sigma S_t dW_t^\mathbb{Q}, \quad S_0 > 0,
\end{equation}
with constant interest rate $r \in \mathbb{R}_+$ and volatility $\sigma > 0$. 
The solution is the well-known geometric Brownian motion:
\begin{equation}
S_t = S_0 \exp \left( \left( r - \frac{1}{2} \sigma^2 \right) t + \sigma W_t^\mathbb{Q} \right).
\end{equation}
Given we are interested in learning the distribution of the payoffs; we now define the payoff functional of path-dependent option.
Let $f : C([0,T],\mathbb{R}_+) \to \mathbb{R}_+$ be a measurable functional mapping a sample path $\omega \mapsto \{S_t(\omega)\}_{t\in[0,T]}$ to a terminal payoff.
Among an interesting class of such options, we choose to work on (arithmetic) Asian call options; hence
\begin{equation}\label{eq:payoff}
f(S_{0:T}) = \max \left( \frac{1}{T} \int_0^T S_t dt - K, 0 \right),
\end{equation}
with $K$ denoting the strike price, and namely the random variable of interest becomes $Y := e^{-rT}f(S_{0:T})$.
Our objective is then simply to approximate the distribution $
Z_0 := \mathcal{L}^{\mathbb{Q}}(Y \mid \mathcal{F}_0),$ rather than merely its expectation.

 ~

\subsection{Path-Augmented Representation}

In, for instance, option pricing via Monte Carlo simulation, we are interested in the computing the expected payoff, i.e. $\E [f(S_{0:T})]$.
In our model, however, we aim to learn the entire conditional distribution of the payoffs given the \emph{filtration} $\mathcal{F}_t$ at $t$.
The filtration $\mathcal{F}_t:=\sigma(S_u:u\leq t)$ defines the all path history $\{S_u\}_{u\in[0,t]}$, which takes values in
the infinite-dimensional Banach space $C([0,t],\mathbb{R})$ endowed with the supremum norm.
Hence, conditioning on $\mathcal{F}_t$ is equivalent to conditioning on an infinite-dimensional object, possibly posing computational challenges.
To address and most importantly overcome this, we define a finite-dimensional adapted process $s_t\in\mathbb{R}^d$ that summarizes all the information in regard to the payoffs, thus enabling a meaningful and practical recursive modeling.
To enable recursive estimation of the value distribution, we formulate a Markov state embedding of the path:
%
%



\begin{proposition}[Finite-Dimensional Sufficient State for Path-Dependent Payoffs]\label{prop:s_t}
	Let $f : C([0,T],\mathbb R_+) \to \mathbb R$ be a measurable payoff functional of the price path $S_{0:T}$.
	For $t\in[0,T]$, suppose there exists an $\mathcal F_t$--measurable, finite-dimensional summary
	$s_t = \psi_t(S_{0:t}) \in \mathbb R^d$ such that for every bounded Borel $h:\mathbb R\to\mathbb R$,
	\[
	\mathbb E^\mathbb{Q}\!\bigl[h(f(S_{0:T})) \mid \mathcal F_t \bigr]
	= g_{t,h}(s_t)
	\]
	for some measurable $g_{t,h}:\mathbb R^d\to\mathbb R$. Then the conditional law of the payoff given $\mathcal F_t$ depends on the past only through $s_t$; i.e.
	\[
	\mathcal L^\mathbb{Q}\!\bigl(f(S_{0:T}) \mid \mathcal F_t \bigr)
	= \mathcal L^\mathbb{Q}\!\bigl(f(S_{0:T}) \mid s_t \bigr).
	\]
\end{proposition}

\begin{proof}
	Because $s_t$ is $\mathcal F_t$--measurable we have $\sigma(s_t)\subseteq\mathcal F_t$.
	By assumption, for indicator functions $h_y(x)=\mathbf 1_{\{x\le y\}}$ (rational $y$),
	\[
	\mathbb{Q}\!\bigl(f(S_{0:T})\le y \mid \mathcal F_t\bigr)
	= g_{t,h_y}(s_t),
	\]
	so the conditional CDF given $\mathcal F_t$ is a measurable function of $s_t$ alone.
	Hence the conditional distribution given $\mathcal F_t$ equals that given $s_t$.
\end{proof}


%
%
Proposition \ref{prop:s_t} justifies the recursive Bellman update in distributional Reinforcement Learning, soon we discuss in detail.
To our knowledge, this formalization of state-based measurability for path-dependent payoffs, using a Doob-Dynkin lemma, has not appeared in the distributional RL literature and is adapted here to establish the theoretical validity of our reduced representation $s_t$ for recursive learning.
Proposition \ref{prop:s_t} is especially important where estimation or updating recursively relies on conditioning upon the state rather than full path, or the path alone.
However, there is a practical result of Proposition \ref{prop:s_t} that is central in a financial setting as ours.
For Asian options, the payoff defined \eqref{eq:payoff} could be also stated as $f(\cdot)= \max (A_T -K, 0)$.
The idea behind our rationale is quite straightforward, as the payoff depends on $A_T$ that evolves from $S_t$; we put forward that the state could be represented as 
\begin{equation}
	s_t = (S_t, A_t, t),
\end{equation} 
for which we implicitly condition upon in Proposition \ref{prop:s_t}. 
Furthermore, we from this point on refer to the conditional distribution of the payoffs as a function of $s_t$.
%
%
We now  define the value distribution process $\{Z_t\}_{t\in[0,T]}$ such that
$Z_t := \mathcal{L}^\mathbb{Q} \left( e^{-r(T-t)} f(S_{[t,T]}) \mid s_t \right)$.
That is, $Z_t$ denoting the conditional distribution of the (discounted) payoff given the finite-dimensional Markovian state $s_t$.
This modeling perspective actually highlights the characteristics of DistRL and is a step from classical value based modeling to distributional modeling.
The process $\{Z_t\}_{t \in [0,T]}$ evolves randomly over time as the governing state changes, thus (hypothetically) capturing the  dynamics of the market.
This  representation is  useful in financial applications where understanding the variability, skewness, or tail behavior of potential outcomes is as important as estimating their mean.
Formally, each $Z_t$ is a $\mathcal{P}(\mathbb{R})$-valued random variable adapted to the filtration generated by $\{s_u\}_{u\leq t}$. 
The mapping $s_t\mapsto Z_t$ defines a distribution-valued process if the transition dynamics of $s_t$ and the terminal functional $f$ are compatible with the Markov property.
This perspective enables the use of operator-based methods such as distributional Bellman operators to propagate distributions over time.

~


\begin{definition} [Wasserstein Metric, from \cite{dabney2017distributional}]
	Let $W_p(\mu, \nu)$ denote the $p$-Wasserstein distance between probability measures $\mu, \nu \in \mathcal{P}_p(\mathbb{R})$, defined as:
	$$W_p^p(\mu, \nu) := \inf_{\gamma \in \Pi(\mu, \nu)} \int_{\mathbb{R} \times \mathbb{R}} |x - y|^p d\gamma(x, y),$$
	where $\Pi(\mu, \nu)$ is the set of all couplings of $\mu$ and $\nu$.
\end{definition}
While we do not explicitly model upon this distance, we note that the distributional Bellman operator is known to be a contraction under the $W_1$ metric in classical settings \cite{bellemare2017distributional}.
This further suggests that the recursive updates of the value distributions are stable under weak convergence criteria as induced by the Wasserstein metric.
In this setting, even approximate representations of $Z $ converge under mild assumptions, providing a principled foundation for reliable distributional learning in quantitative finance, \cite{rowland2018analysis}, \cite{bellemare2017distributional}.

We conclude that minimizing the Wasserstein distance alone may not suffice to evaluate the practical adequacy of the learned value distributions for option pricing, as it does not directly reflect financially critical aspects such as tail behavior or extreme quantile accuracy.
Given that no thorough study of this limitation exists to our knowledge, we retain this perspective in our modeling approach and leave a deeper investigation to future research; especially its relevance for path-dependent option pricing remains an open question.
\quad
\begin{definition}[Distributional Bellman Operator, adapted from \cite{dabney2017distributional}]\label{distBO}
	Let $\mathcal{T} $ be the distributional Bellman operator acting on a distribution-valued function $ Z : \mathcal{S} \to \mathcal{P}_p(\mathbb{R}) $, where $ \mathcal{P}_p(\mathbb{R}) $ is the set of probability measures on $ \mathbb{R} $ with finite $ p $-th moment. For a given state $ s_t \in \mathcal{S} $, we define
	\[
	(\mathcal{T}Z)(s_t) := \mathcal{L}^\mathbb{Q} \left[ R_t + \gamma Y_{t+1} \mid s_t \right], \quad \text{where } Y_{t+1} \sim Z(s_{t+1}),\ s_{t+1} \sim P(\cdot \mid s_t),
	\]
here $ R_t $ as the immediate reward, and $ \gamma = e^{-r\Delta t} $ as the discount factor per time step. The randomness is over the transition $ s_{t+1} $ and the sampled future payoff from $ Z(s_{t+1}) $.
\end{definition}
\quad

The distributional Bellman operator was introduced by \cite{dabney2017distributional}; we here specialize this operator for the Markovian state-only setting $Z(s_t)$ rather than the state action $Z(s,a)$ and we incorporate the fact that we operate under risk-neutral probability measure with a financial setting.
While the current formulation is applied to a single underlying asset, the framework readily generalizes to multi-asset settings.
For example, pricing a basket option under joint dynamics.

\begin{remark}
Unlike classical reinforcement learning settings, our formulation does not involve an agent making sequential decisions through actions. 
That is, the state space $\mathcal{S}$ evolves purely randomly, and no policy $\pi$ is present to modulate transitions or rewards.
As such, the (distributional) Bellman operator $\mathcal{T}$ we here adopt corresponds to the uncontrolled case, where the dynamics are fully governed by the exogenous transition kernel $P(s_{t+1} \mid s_t)$ and fixed reward structure $R_t$.
This simplification is, we remark, reasonable in   option pricing, where the objective is not to choose actions, but to determine the distribution of future payoffs under a specified stochastic model.
\end{remark}

While the contraction property of the distributional Bellman operator is well-established in the reinforcement learning literature, its financial implications remain underexplored.
With that in mind in our context, we aim to provide rigorous foundation for recursively estimating the full distribution of discounted payoffs which is a key requirement in pricing exotic options where mean-based methods might  obscure important tail behavior. 
The ability to  propagate distributional uncertainty over time could be particularly valuable in financial applications involving risk measures such as Value-at-Risk (VaR), Conditional VaR, or path-dependent regulatory capital requirements.
%
%
%
%
%
%
All the complexity aside, our goal is to approximate the  value distribution $ Z $ by learning a distribution-valued function:
\[
\Phi : \mathcal{S} \to \mathcal{P}_p(\mathbb{R}),
\]
where $ \mathcal{S} \subset \mathbb{R}^d $ denotes a finite-dimensional, Markovian state space constructed from path-dependent statistics (e.g., current asset level $ S_t $, running maximum $ \sup S_{[0,t]} $, or accumulated average), and $ \mathcal{P}_p(\mathbb{R})$ is the space of Borel probability measures with finite $ p $-th moment. 
This mapping $ \Phi $ interestingly serves as a parametric approximation to the value distribution process $ Z $, and is learned recursively via the application of a distributional Bellman operator. 
We now formalize the functional setting and establish the key contraction property that ensures convergence to a unique fixed-point distribution.

\begin{proposition}[Distributional Bellman Contraction, adapted to measurable-state notation from \cite{bellemare2017distributional}]
	Let $(\mathcal S,P,r,\gamma)$ be a Markov decision process with measurable state space
	$\mathcal S$, transition kernel $P(\cdot\mid s)$, discount $0<\gamma<1$,  
	and reward function $r:\mathcal S\times\mathcal S\to\mathbb R$ such that  
	$\sup_{s\in\mathcal S}\mathbb E_{s'\sim P(\cdot\mid s)}\!\bigl[|r(s,s')|^p\bigr]<\infty$  
	for some $p\!\in\![1,\infty)$.  
	Define
	\[
	\mathcal Z := \bigl\{Z:\mathcal S\to\mathcal P_p(\mathbb R)\bigr\},\qquad
	W_p(Z_1,Z_2)\;:=\;\sup_{s\in\mathcal S} W_p\!\bigl(Z_1(s),Z_2(s)\bigr),
	\]
	where $W_p$ is the order-$p$ Wasserstein distance on $\mathcal P_p(\mathbb R)$.  
	For $Z\in\mathcal Z$, define
	\[
	(\mathcal T Z)(s) \stackrel{d}{=} r(s,S')+\gamma\,Y,
	\qquad S'\sim P(\cdot\mid s),\; Y\sim Z(S') \text{ (cond.\ on $S'$)}.
	\]
	Then $\mathcal T$ is a $\gamma$-contraction on $(\mathcal Z,W_p)$:
	\[
	\forall Z_1,Z_2\in\mathcal Z:\qquad
	W_p\!\bigl(\mathcal T Z_1,\mathcal T Z_2\bigr)
	\;\le\;\gamma\,W_p(Z_1,Z_2).
	\]
	Consequently $\mathcal T$ has a unique fixed point $Z^\star\in\mathcal Z$.
\end{proposition}

\begin{proof}
	Fix $s\in\mathcal S$ and couple both copies of the operator on the same transition:
	\[
	S'\sim P(\cdot\mid s),\qquad
	(Y_1,Y_2) \text{ an optimal $W_p$-coupling of } Z_1(S'),Z_2(S').
	\]
	Let $R := r(s,S')$ and set $X_i := R + \gamma Y_i$, $i=1,2$, so $X_i \sim (\mathcal T Z_i)(s)$.
	Because $R$ is common to both $X_1,X_2$,
	\[
	|X_1-X_2|^p = \gamma^p |Y_1-Y_2|^p.
	\]
	Taking expectations and using optimality of the conditional coupling,
	\[
	W_p^p\!\bigl((\mathcal T Z_1)(s),(\mathcal T Z_2)(s)\bigr)
	\;\le\;
	\gamma^p\,\mathbb E_{S'\sim P(\cdot\mid s)}
	\Bigl[ W_p^p\!\bigl(Z_1(S'),Z_2(S')\bigr) \Bigr].
	\]
	Hence
	\[
	W_p\!\bigl((\mathcal T Z_1)(s),(\mathcal T Z_2)(s)\bigr)
	\le
	\gamma\,\Bigl(
	\mathbb E_{S'\sim P(\cdot\mid s)}
	\bigl[ W_p^p(Z_1(S'),Z_2(S')) \bigr]
	\Bigr)^{1/p}
	\le \gamma\,\sup_{s'\in\mathcal S}W_p\!\bigl(Z_1(s'),Z_2(s')\bigr).
	\]
	Taking the supremum over $s$ gives the stated contraction.
	Finally, the reward $p$-moment bound together with $Z\in\mathcal Z$ implies
	$(\mathcal T Z)(s)\in\mathcal P_p(\mathbb R)$ for all $s$, so $\mathcal T:\mathcal Z\to\mathcal Z$.
	Since $\mathcal P_p(\mathbb R)$ is complete under $W_p$ and the sup metric over a family of
	complete spaces is complete, Banach's fixed-point theorem yields a unique $Z^\star$.
\end{proof}

\cite{bellemare2017distributional} establishes the same $\gamma$--contraction result for the
distributional Bellman (policy-evaluation) operator via an optimal coupling argument.
Our presentation restates the theorem in a measurable-state setting and makes the coupling step implicit.
Because the same sampled reward $r(s,S')$ appears in both arms, the \emph{shift invariance} of
$W_p$ cancels the reward term, and \emph{positive homogeneity} yields the $\gamma$ factor.
Taking expectation over the transition and bounding by a supremum gives the contraction.
This form is convenient for continuous-state financial models derived from SDEs; 
path-dependent payoffs are handled by augmenting the state with the relevant path
functional so that the process remains Markov.

In practical terms, the $\gamma$-contraction of the distributional Bellman operator in the Wasserstein metric ensures that errors in the learned distribution at one time step do not amplify through recursion. This provides robustness to estimation noise, discretization error, or model misspecification; all of which are prevalent in financial settings.

In contrast, classical pricing methods that rely on recursive expectations (e.g., dynamic programming with expected value functions) may suffer from moment drift or tail distortion when propagated backward in time. The contraction in the full return distribution guarantees that not only the mean, but the entire shape of the distribution (including tails and higher moments), evolves in a stable and controlled manner, \emph{in theory}.
The recursive contraction of the value distribution in Wasserstein space could be thought  as a form of regularization across time.



\subsection{Reduction to Path-Dependent Option Pricing}

We now instantiate the abstract framework for the specific case of pricing path-dependent options under the risk-neutral measure. 
Consider an Asian option, where the payoff depends on the running average of the asset price.
We define a discretized Markovian state vector:
\begin{equation}
s_t = (S_t, A_t, t) \in \mathbb{R}^3
\end{equation}
where  $S_t$  is the spot price, $ A_t = \frac{1}{t} \int_0^t S_u du $ is the running average, and $ t $ is the current time.
We define the reward structure as:
\begin{equation}
R_t = 0 \quad \text{for } t < T, \quad \text{and } R_T = f(S_{0:T}),
\end{equation}
so that the entire payoff is realized at terminal time $T$.
Then, the terminal value distribution is deterministic:
\[
Z_T(s_T) = \delta_{f(S_{0:T})}
\]
and the backward recursion becomes:
\[
Z_t(s_t) = \mathcal{L} \left[ \gamma Z_{t+1}(s_{t+1}) \mid s_t \right]
\]

This formulation reduces the problem of pricing a path-dependent option to a  recursive distributional estimation problem.
At each time step, we propagate a full distribution of future payoffs conditioned on the Markovian path-augmented state $ s_t $. 
We approximate the value distribution $ Z_t(s_t)$ by a finite collection of quantiles $ \{\theta_i(s_t)\}_{i=1}^N $, trained using quantile regression.
Each quantile estimate $ \theta_i $ is updated via a stochastic gradient of the quantile loss function. To ensure the validity of these updates, we recall the standard result that the quantile regression loss admits an unbiased stochastic gradient estimator:

\begin{proposition}[Unbiased Stochastic Gradient for Quantile Loss, from \cite{dabney2017distributional}]\label{prop:quantile_unbiased}
	Let the quantile loss be defined as:
$$\rho_{\tau_i}(u) = u(\tau_i - \mathbb{I}_{\{u < 0\}}) \quad \text{where } u = r + \gamma \theta_i(s') - \theta_i(s)$$
Then the stochastic gradient:
$$
g_i(s, s') := \nabla_{\theta_i(s)} \rho_{\tau_i}(r + \gamma \theta_i(s') - \theta_i(s)) = -(\tau_i - \mathbb{I}_{\{r + \gamma \theta_i(s') - \theta_i(s) < 0\}})$$
is an unbiased estimator of the gradient of the expected quantile loss:
$$\mathbb{E}_{(s, s')} [g_i(s, s')] = \nabla_{\theta_i(s)} \mathbb{E}_{(s, s')} \left[ \rho_{\tau_i}(r + \gamma \theta_i(s') - \theta_i(s)) \right]$$
	
\end{proposition}

\begin{proof}
We differentiate the quantile loss:
$$
\frac{\partial}{\partial \theta_i(s)} \rho_{\tau_i}(r + \gamma \theta_i(s') - \theta_i(s)) = -(\tau_i - \mathbb{I}_{\{r + \gamma \theta_i(s') - \theta_i(s) < 0\}})$$
Then take the expectation over transitions $(s, s') \sim D$:
$$\mathbb{E}_{(s, s')} \left[ \frac{\partial}{\partial \theta_i(s)} \rho_{\tau_i}(r + \gamma \theta_i(s') - \theta_i(s)) \right] = \frac{\partial}{\partial \theta_i(s)} \mathbb{E}_{(s, s')} \left[ \rho_{\tau_i}(r + \gamma \theta_i(s') - \theta_i(s)) \right],$$
proving the unbiasedness of the stochastic gradient estimator $g_i$.
\end{proof}

This result ensures that the updates used in learning the quantile representations of \( Z_t \) are theoretically sound.
In our implementation, these gradients are used to recursively estimate the distribution of option payoffs via stochastic gradient descent, under the recursive structure induced by the distributional Bellman operator.

\begin{remark}
	In our DistRL framework for Asian options, we minimize the quantile loss over a continuous state space using RBF approximators. 
	The stochastic gradients used to update the quantile parameters are valid and unbiased, as guaranteed by Proposition~\ref{prop:quantile_unbiased}, even under function approximation.
\end{remark}

\begin{remark}[Quantile Consistency]
	Under standard assumptions (e.g., correct sampling of transitions or sufficient data), minimizing the quantile regression loss ensures that each parameter \(\theta_i(s)\) converges to the \(\tau_i\)-quantile of the Bellman target distribution:
	\[
	\theta_i(s) \to \text{Quantile}_{\tau_i} \left( R(s) + \gamma Z(s') \right)
	\]
	This alignment guarantees that the learned distribution \(Z(s)\) retains a consistent probabilistic interpretation across updates. In financial applications, this is particularly important for quantifying tail risk, pricing, and scenario analysis.
\end{remark}
Remark above guarantees that the learned quantile functions converge to fixed-points of the Bellman operator applied to the cumulative average $A_t$. 
In our case, the state includes $(S_t, A_t, t)$ which ensures the necessary Markov structure for Bellman recursion.
Following \cite{dabney2017distributional}, we approximate the value distribution $ Z(s) $ by its inverse cumulative distribution function evaluated at fixed quantile levels:
\[
\tau_i := \frac{i - 0.5}{N}, \quad \hat{Z}(s) = \frac{1}{N}\sum_{i=1}^N \delta_{\theta_i(s)}
\]
where $ \theta_i(s) \in \mathbb{R} $represents the estimate of the $ \tau_i $-quantile of $ Z(s) $, and $ \delta_{\theta_i(s)} $ denotes a Dirac measure centered at $ \theta_i(s) $. 
This construction provides a discrete distributional representation of $ Z(s)$ with empirical support.
The parameters $ \theta_i(s) $ are learned by minimizing the quantile regression loss:
\[
\mathcal{L}_{\mathrm{QR}} = \sum_{i=1}^N \mathbb{E}_{(s, s')} \left[ \rho_{\tau_i} \left( r(s) + \gamma \theta_i(s') - \theta_i(s) \right) \right]
\]
where the quantile loss is given by:
\[
\rho_{\tau}(u) = u \cdot (\tau - \mathbb{I}_{\{u < 0\}}.)
\]
We interpret $ \hat{Z}(s)$ as a quantized approximation of the true return distribution \( Z(s) \), supported on $ \{ \theta_i(s) \}_{i=1}^N$, with each $ \theta_i(s) \approx F_Z^{-1}(\tau_i)$.
This Dirac-based representation is especially amenable to recursive learning and tractable propagation under distributional Bellman updates.
The discrete quantile-based representation put forward enables stable stochastic updates and avoids parametric assumptions about the shape of the value distribution.
Moreover, it aligns naturally with the financial interpretation of quantiles as Value-at-Risk (VaR) levels, which is particularly useful for path-dependent option pricing.
%
%
%
%
%
We interpret $ \hat{Z}(s) $ as an empirical approximation to the true return distribution $ Z(s) $, supported on $ \{\theta_i(s)\}_{i=1}^N $, where each $ \theta_i(s) \approx F_Z^{-1}(\tau_i) $.
This construction introduces Dirac deltas for the first time in the approximation layer.

\begin{remark}[Dirac Approximation Layer]
	The use of Dirac measures $\delta_{\theta_i(s)}$ constitutes an empirical approximation to the cumulative distribution function (CDF) of $Z(s)$ via its inverse. 
	This is the only stage where measure-valued representations are discretized, enabling (analytically) tractable gradient updates and convergence analysis.
\end{remark}


Under standard regularity assumptions on the state space, function class, and sampling procedure, the following result guarantees convergence of the quantile parameters $\theta_i(s) $ to the true quantiles of the return distribution such as

\begin{proposition}[Quantile SGD Converges to Population Quantiles; after\footnote{Adapted and generalized to continuous state space and financial payoff distributions} \cite{dabney2017distributional,dabney2018implicit}]\label{prop_suff_data}
	Let $(\Omega,\mathcal F,Q)$ be a probability space and let $Z(s)$ denote the (discounted) return distribution from a Markov state $s \in \mathcal S \subset \mathbb R^d$ (the state $s$ aggregates any path-dependent features).
	Fix a quantile level $\tau_i \in (0,1)$.
	Assume that for the fixed state $s$:
	
	\begin{enumerate}
		\item $\mathbb E^Q[|Z(s)|^2] < \infty$.
		\item The CDF $F_s$ of $Z(s)$ is continuous at its $\tau_i$-quantile $q_i(s)$ (so the minimizer is unique).
		\item We observe an i.i.d.\ sample $\{Z_k(s)\}_{k\ge 0}$ from $Z(s)$.
		\item Step sizes $\{\alpha_k\}$ satisfy $\sum_k \alpha_k = \infty$, $\sum_k \alpha_k^2 < \infty$.
	\end{enumerate}
	
	Consider the stochastic (sub)gradient iteration
	\[
	\theta_i^{(k+1)}(s)
	= \theta_i^{(k)}(s)
	+ \alpha_k \bigl(\tau_i - \mathbf 1_{\{Z_k(s) < \theta_i^{(k)}(s)\}}\bigr),
	\]
	which is equivalent to stochastic gradient descent on the quantile (pinball) loss
	$\rho_{\tau_i}(u) = u(\tau_i - \mathbf 1_{\{u<0\}})$ with $u = Z_k(s) - \theta_i^{(k)}(s)$.
	
	Then $\theta_i^{(k)}(s) \to q_i(s)$ almost surely (hence in probability), where
	\[
	q_i(s) = \inf\{x : F_s(x) \ge \tau_i\}
	= \arg\min_\theta \mathbb E^Q[\rho_{\tau_i}(Z(s) - \theta)].
	\]
\end{proposition}

\begin{proof}
	Write $U_k = Z_k(s) - \theta_i^{(k)}(s)$ and $g(U_k) := \tau_i - \mathbf 1_{\{U_k<0\}}$.
	Then the update is $\theta_{k+1} = \theta_k + \alpha_k g(U_k)$.
	Conditional on $\theta_k$, $\mathbb E[g(U_k) \mid \theta_k] = \tau_i - F_s(\theta_k)$.
	Define the mean field $h(\theta) := \tau_i - F_s(\theta)$.
	Assumption~(2) implies $h$ is continuous and strictly positive below $q_i(s)$, negative above, and zero at $q_i(s)$.
	Thus $q_i(s)$ is the unique root of $h$ and is a globally asymptotically stable equilibrium of the associated ODE $\dot \theta = h(\theta)$.
	The noise sequence $M_{k+1} := g(U_k) - h(\theta_k)$ is a martingale difference with bounded second moment by Assumption~(1).
	Standard stochastic approximation results  yield $\theta_k \to q_i(s)$ a.s.
	Finally, note $\mathbb E[\rho_{\tau_i}(Z(s)-\theta)]$ is minimized exactly at $q_i(s)$ by properties of the pinball loss, giving the last equality.
\end{proof}

This result ensures that, given sufficient training data and appropriate learning rates, the quantile estimates \( \theta_i(s) \) converge to the true inverse CDF values at levels \( \tau_i \), thereby allowing an accurate reconstruction of the value distribution \( Z(s) \). 
Although we use a finite-dimensional RBF basis rather than a neural network, the convergence result still applies under the same boundedness and step-size conditions.
This validates our training procedure empirically and theoretically.

	
	

\subsection{From Distributional Bellman Equations to Quantile Optimization}

While the distributional Bellman operator provides a recursive definition of the value distribution, direct implementation is infeasible. 
To navigate through the recursion, we minimize the difference between the predicted and actual quantile values via quantile loss as in \cite{dabney2017distributional}:
\begin{equation}\label{eq:l_qr}
\mathcal{L}_{\mathrm{QR}} = \sum_{i=1}^N \mathbb{E} \left[ \rho_{\tau_i}(r + \gamma \theta_i(s') - \theta_i(s)) \right],
\end{equation}
with \eqref{eq:l_qr}, the core of the learning problem is formed. 
In our implementation, we approximate each quantile function $ \theta_i(s) $ using a fixed radial basis function (RBF) expansion $\phi(s) \in \mathbb{R}^d $, such that
\begin{equation}\label{eq:rbf1}
\theta_i(s) = w_i^\top \phi(s).
\end{equation}
The parametrization in \eqref{eq:rbf1} transforms the learning of quantile functions into the learning of linear weights $ w_i \in \mathbb{R}^d $ over non-linear features. 
We now provide a proposition for the closed-form stochastic gradient update of quantile regression when using RBF function approximators.
While inspired by the quantile loss framework of \cite{dabney2017distributional}, our derivation is specific to our parametric setting and gives the explicit form of the stochastic gradient update used in this setting, that is

\begin{proposition}[Stochastic (Sub)Gradient for RBF-Based Quantile TD Loss]
	Let $\phi(s)\in\mathbb R^d$ be a fixed feature vector and approximate the $i$th quantile by
	$\theta_i(s)=w_i^\top\phi(s)$.
	Given a transition $(s,r,s')$ and discount $\gamma\in[0,1)$ define the TD residual
	\[
	\Delta_i = r + \gamma\,\theta_i(s') - \theta_i(s)
	= r + \gamma\,w_i^\top\phi(s') - w_i^\top\phi(s).
	\]
	Let the quantile (pinball) loss be $\rho_{\tau_i}(u)=u(\tau_i-\mathbf 1_{\{u<0\}})$, with subgradient
	$c_i := \tau_i - \mathbf 1_{\{u<0\}}$ (any selection at $u=0$).
	Then the \emph{full} parameter (sub)gradient is
	\[
	\nabla_{w_i}\rho_{\tau_i}(\Delta_i)
	= (\tau_i - \mathbf 1_{\{\Delta_i<0\}})\,(\gamma \phi(s') - \phi(s)).
	\]
	If we adopt the standard temporal-difference \emph{semi-gradient} that stops gradients
	through the bootstrapped target $r+\gamma \theta_i(s')$, we obtain
	\[
	\nabla_{w_i}^{\mathrm{semi}} \rho_{\tau_i}(\Delta_i)
	= -(\tau_i - \mathbf 1_{\{\Delta_i<0\}})\,\phi(s).
	\]
	An SGD update with learning rate $\eta>0$ is then
	\[
	w_i \leftarrow w_i + \eta\,(\tau_i - \mathbf 1_{\{\Delta_i<0\}})\,\phi(s)
	\]
	(semi-gradient case).
\end{proposition}

\begin{proof}
	Write $c_i := \tau_i - \mathbf 1_{\{\Delta_i<0\}}$ (choose any subgradient at $\Delta_i=0$).
	By the chain rule,
	$\nabla_{w_i}\rho_{\tau_i}(\Delta_i) = c_i \nabla_{w_i}\Delta_i$.
	Since $\theta_i(s)=w_i^\top\phi(s)$, $\theta_i(s')=w_i^\top\phi(s')$,
	we have $\nabla_{w_i}\Delta_i = \gamma\phi(s') - \phi(s)$, giving the full gradient.
	If gradients are not propagated through the target term
	$r+\gamma\theta_i(s')$, then
	$\nabla_{w_i}\Delta_i = -\phi(s)$, yielding the stated semi-gradient.
\end{proof}

In practice, we apply the \emph{semi-gradient} variant of the quantile loss, where the bootstrapped target $\theta_i(s')$ is treated as a constant during differentiation. This avoids propagating gradients through the next-state estimate and improves numerical stability, as is standard in temporal-difference learning, \cite{sutton2018reinforcement}. 
Formally, although the quantile loss depends on both $\theta_i(s)$ and $\theta_i(s')$, we compute:
\[
\nabla_{w_i} \rho_{\tau_i}(r + \gamma \theta_i(s') - \theta_i(s)) \approx -\left(\tau_i - \mathbb{I}_{\{r + \gamma \theta_i(s') - \theta_i(s) < 0\}}\right) \cdot \phi(s)
\]
This form highlights how the gradient combines the TD residual \( \Delta_i \), the RBF features \( \phi(s) \), and the asymmetry of the quantile loss. Unlike traditional TD learning, the update direction depends on the quantile level \( \tau_i \), thereby adjusting each quantile independently.
In practice, these updates are performed in parallel for all \( i = 1, \dots, N \), enabling a full approximation of the value distribution \( \hat{Z}(s) \) at every visited state.

The distributional reinforcement learning formulation developed above is particularly well-suited for financial applications where understanding the full distribution of outcomes is more informative than merely computing expected values.
In the context of path-dependent options, such as Asian  options, the shape of the terminal value distribution encodes valuable information.
Furthermore, the ability to approximate the value distribution using state-dependent quantiles allows the model to dynamically react to evolving market states, including sudden jumps or volatility shifts.
For instance, the quantile estimates $ \theta_i(s) $ may show sharp responses in regimes characterized by high uncertainty or after discontinuous moves in the underlying process. 
These adaptive updates could rigorously  offer a mechanism for capturing distributional shifts without explicitly specifying various components in the model.

The payoff functional $f(S_{0:T}) $, which ultimately determines the terminal distribution $Z_T $, often exhibits rich statistical structure.
For instance, in Asian options, it involves a time-averaging operator introducing    complex distributional features.
While fully acknowledging the challenges associated with our reinterpretation; we aim to contribute given our expectation of more and more \emph{model-free} approaches surfacing in the coming years.
We now move on to numerical illustration to provide practical implications of the model presented in detail.

\section{Numerical Illustrations}

Training alone on no-prior knowledge of the distribution shape carries its own risk.
To provide a stable and neutral starting point, we initialize  weight vector of each quantile uniformly so that the resulting predicted quantiles  at the initial state approximates the mean payoff from a Monte Carlo simulation. 
This creates a degenerate distribution (i.e., all quantiles are initially identical), centered at the mean payoff.
Although we remark uninformative nature of the distributional shape, this initialization offers a stable baseline  to learn the true distributional characteristics through training.
Another major reason is the possibility of early gradient explosion.
Such incidents could occur if initial weights are far away from the target, quantile loss gradients might be larger than expected.
This aspect might introduce large updates hence destabilizing the learning process.
As training progresses, the model is expected to learn to spread the quantiles to match the true distributional shape.
In financial context; we aim to prevent the model from initially over- or under-pricing severely while learning distributional aspects.
As in a setting like ours, random initialization might produce extreme or unreasonable values so that we strictly caution against it.

During each epoch, we have samples a batch of stochastic trajectories (episodes) under the risk-neutral asset price dynamics.
We denote the number of paths in an epoch by $N_{\text{path}}$, and each path in an epoch starts at a different seed of randomness\footnote{We prefer to work with np.random.seed(.)}.
Each episode proceeds from the initial state through discrete time steps until maturity.
States along the path consist of the current spot price, running average, and time index.
The terminal payoff is computed and used as a training target to update the quantile value estimate at the final state.
Subsequently, recursive backward updates are applied to intermediate states via the distributional Bellman equation.
For each episode, we compute the TD-targets backward from terminal state to the initial state using this recursive formulation
Quantile weights are then updated via the quantile regression gradient with learning rate 
$\eta$, using clipped gradient $l_2$-norm to ensure stability aligned with \cite{dabney2017distributional} though we differ by not using Huber loss.

To represent the quantile estimates $ \theta_i(s) $, we use a radial basis function (RBF) feature map defined over the normalized state space. 
Each state $ s = (S_t, A_t, t) $ is normalized to the unit cube $ [0, 1]^3 $ by dividing each component by fixed scaling constants.
The feature vector $ \phi(s) \in \mathbb{R}^{d} $ includes a bias term and $ n$ Gaussian RBFs centered at locations $ \{c_j\}_{j=1}^n \subset [0,1]^3 $, with fixed bandwidth $ \sigma = 0.5 $.
Formally, the RBF features are given by:
\begin{equation}\label{eq:RBF}
\phi_j(s) = \exp\left( -\frac{\|s - c_j\|^2}{2\tilde{\sigma}^2} \right), \quad j = 1, \dots, n,
\end{equation}
thus $\phi(s) = \left[1, \phi_1(s), \dots, \phi_n(s)\right]^\top$.
We use $ n = 40 $ RBF centers, sampled uniformly at random from $ [0, 1]^3 $. 
This setup provides a smooth and localized function approximation basis suitable for modeling the nonlinear structure of the option payoff distribution.
The number of RBF centers and $\tilde{\sigma}=0.5$ are fixed throughout the study as we have no interest in further increasing the computational complexity.
Another major point we remark is the state vector in \eqref{eq:RBF}.
So as to keep the RBF input domain within $[0,1]$ and further match the domain over which RBF centers are randomly sampled; we normalize the state vector, $s$, by 200 given our initial price so it never gets greater than one, except the number of steps in the state vector.
Surely, more sophisticated approaches could be offered, yet for our scope we remark the sufficiency of such simplicity.
However, normalizing by 200, in our case, is simple, stable and could be easily justified given the price range in our simulations. 
During training, we employ a Monte Carlo reference set of payoffs based on 3000 simulated  paths using \eqref{eq:sde} to monitor convergence diagnostics per epoch. 
This considerably smaller but consistent sample provides a stable validation reference across epochs without excessive computational burden. 
This reference set is not used for training but solely to evaluate convergence of the model.
For final evaluation, we independently generate a larger  sample of 100,000 paths to serve as a statistically robust benchmark. 
This allows us to report final accuracy metrics, including mean payoff, distributional distances, and quantile prediction error, against a high-precision estimate of the payoff distribution.
We note that we employ 100 paths for 100 epochs, and only the 100-path (of the first epoch) is using the same seed of randomness with the larger benchmark set.
The reason we have an independent set as a benchmark is to see the performance on out-of-sample rather than in-sample; so as to see if our suggestion could handle generalization to some degree.
We restate that once trained; the method should be able to used later with a degree of accuracy and be able to provide consistent answers.
This marks the cornerstone in an any trustworthy approach of option pricing.
Each training epoch, we generate fresh paths (100 per epoch) with random seeds (no reset between epochs), so the trajectories are new every time, and as previously stated we have 100 epochs.
Although we observe drastic improvement in many cases\footnote{since we carry forward the learned parameters to the next epoch, rather than resetting	and that is what enables the distributional recursion to converge over time} with intensely higher number of such training over epochs; this perspective might not be possible at all given the availability of data in financial setting (i.e. the number of derivatives contracts and associated possible payoffs).
Although Proposition \ref{prop_suff_data} states the existence of enough data to train on; given the implausibility in our field, we report our finding on such training set as opposed to many machine learning data requirements.

Our report in this part rather involves two parts: experiments in a reasonably realistic financial settings and reasonably unrealistic financial settings to present some aspects that would require a degree of caution.
While not presenting the underlying parameters explicitly for the second part for reasons of clarity and nature of strangely stretched boundaries of financial realities in such experiments; for this first part we keep almost everything constant except the initial price, $S_0$ and provide the parameter space of the experimentation.
We keep the interest rate, $r$, at 0.03 and the volatility, $\sigma$, at 0.2 considering that there exist 252 time steps (i.e. trading day) in a year; the annual volatility of 20 per cent is not far from realistic.
We also set the strike price, $K$, at 100 and keep it constant for all parts.
We keep the number of quantiles in our learning problem at 50, and do not deliberately present a sensitivity analysis in regard to the number of quantiles.
The major reason is the fact that there exist 252 trading days in a year, therefore we could accept it as a reasonable benchmark.
However, we remark that increasing the number of quantiles brings in higher computational burden and during our experiments we failed to observe significant improvement considering the increased time required for computation.

Another major point to consider is the outlier payoffs that could blow up gradient updates, especially in quantile learning, which is highly sensitive to distribution tails.
Clipping the payoff prevents these outliers from destabilizing the learning updates.
Capping as a practical approximation of focusing on reasonable and learnable range  rather than letting a tiny fraction of scenarios dominate training needs no deeper discussion.
Although we will relax this later on.

In Tables \ref{tab:1}, \ref{tab:2} and \ref{tab:3}; comparative results on Absolute Errors are given for learning rate, $\eta$, of 0.005 and Wasserstein Distance is provided only in Table \ref{tab:1}.
Although strike price is unchanged at 100, the initial price varies in each tabled from 5 to 20.
Besides, payoffs are clipped for reason outlined previously. 
In the second section of this manuscript, we remark that deciding along on Wasserstein Distance might be misleading given the complexity of such distributions; therefore we provide a numerical comparison only on Table \ref{tab:1}.
We note that Wasserstein Distance could be sufficient in indicating superior or inferior performance, yet we failed to generalize a benchmark level.
However, we note that absolute errors tend to be lower than one in cases the described distance tends to be lower than one.
We remark that, however, due to the presence of kurtosis and skewness; this becomes a challenging measure to decide upon.
We further note that $S_0 - K$ is definitely not the Asian call option's payoff defined, yet we find it simple and effective enough to create distributions to learn from.

\begin{table}[htbp]
	\centering
	\caption{For   $\eta=0.005$}
	\begin{tabular}{cccccc}
		\\$S_0-K$ & Max Payoff & MC Price & DistRL & Abs. Error & Wasserstein Distance \\
		\midrule
		\multirow{2}[1]{*}{5} & 10    & 5.2510 & 5.6479 & 0.3968 & 0.6094 \\
		& 20    & 7.6511 & 9.9488 & 2.2977 & 2.6017 \\
	\end{tabular}%
	\label{tab:1}%
\end{table}%
In Tables \ref{tab:2} and \ref{tab:3}, we again observe that a reasonable range of payoffs in a narrower field provides a reasonable possibility of learning.
Before we discuss the results, we state that we choose to price a one-year maturity Asian call options so as to keep the distribution wider and difficult learn. 
As long as the difference between the initial price and the strike price is closer (i.e. 5 to 10 per cent on average), our framework provides an acceptable approximation.
However, we stress again that a different set of data is used for final evaluation, hence in case of completely different our framework provides a stable approximation.
Another major point is the requirements of clipped payoffs. 
In all reported results, we observe the absolute errors tend to increase as the training data becomes sparse and too spread out.

\begin{table}[htbp]
	\centering
	\caption{For   $\eta=0.005$}
	\begin{tabular}{ccccc}
		\\$S_0-K$ & Max Payoff & MC Price & DistRL Price & Abs. Error \\
		\midrule
		\multirow{3}[1]{*}{10} & 10    & 6.758 & 5.7523 & 1.0363 \\
		& 20    & 10.5275 & 11.6273 & 1.0997 \\
		& 30    & 12.0681 & 15.8209 & 3.7527 \\
	\end{tabular}%
	\label{tab:2}%
\end{table}%

We note that in all cases, the learning rate is fixed at 0.005, and we remark this is an acceptable value given the updates occurring in the framework.
We argue that higher learning rates, on average, causes serious overestimation (up to 30 per cent on some experiments).
Therefore, in a setting such as ours; any rate between 0.001 to 0.005 is reasonable and is unlikely to blow up in final quantiles.
Yet we stress again this particular choice requires subject related knowledge, and caution must not be spared.

\begin{table}[htbp]
	\centering
	\caption{For   $\eta=0.005$}
	\begin{tabular}{ccccc}
		\\$S_0-K$ & Max Payoff & MC Price & DistRL & Abs. Error \\
		\midrule
		\multirow{3}[1]{*}{20} & 10    & 8.8432 & 6.0662 & 2.7770 \\
		& 30    & 19.3664 & 15.6854 & 3.6810 \\
		& 50    & 21.7369 & 26.2271 & 4.4902 \\
	\end{tabular}%
	\label{tab:3}%
\end{table}%

In out-of-the-money and deep out-of-the-money configurations (i.e. over 10 per cent difference of the initial price and the strike price), the payoff distribution is highly imbalanced, dominated by zero payoffs with very sparse positive values.
This poses interesting challenges for the learning algorithm, because the quantile updates have weak signals.
We suggest boosting path diversity via drift-shifted importance sampling, increasing the training sample size, and initializing quantile estimates to small positive values to avoid collapse. 
These strategies aim to stabilize learning and improve convergence in rare-event regimes.
Given the purpose and the scope of the manuscript, we leave this for a further study.

\begin{figure}[h]
	\begin{subfigure}{0.5\textwidth}
		\centering
		\includegraphics[ width=1\linewidth]{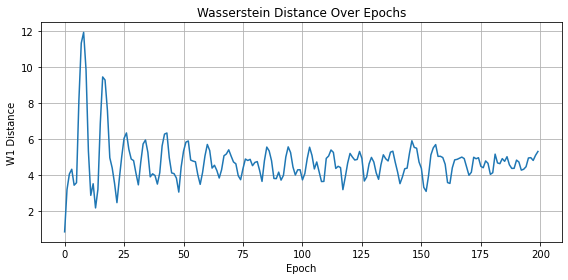}
		\caption{Wasserstein Distance over Epochs}
		\label{plot:W1}
	\end{subfigure}
	\begin{subfigure}{0.5\textwidth}
		\centering
		\includegraphics[width=1\linewidth]{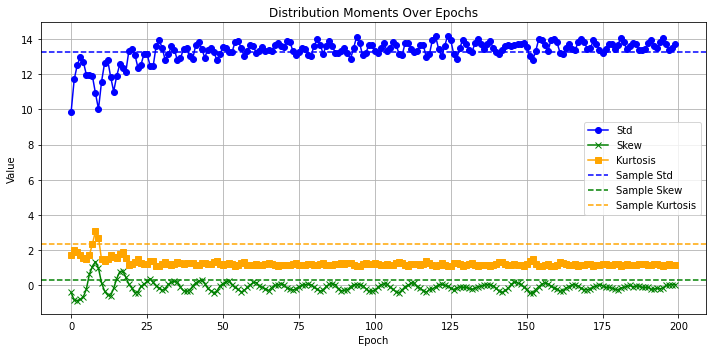}
		\caption{Distribution Moments over Epochs}
		\label{plot:moments}
	\end{subfigure}
	\caption{Counter example of increased epochs }
	\label{plot:200_epocs}
\end{figure}

We stated that more data might be needed for a better learning experience as also theoretically justified in Proposition \ref{prop_suff_data}, and we stand by that while acknowledging the possible difficulty in working with real data.
However, we provide a counter example in Figure \ref{plot:200_epocs} with the number of epochs (trials) doubled.
In Figures \ref{plot:W1} and \ref{plot:moments}, we observe the convergence even before 100 epochs.
In Figure \ref{plot:W1}, we see the learning algorithm learns considerably well.
In Figure \ref{plot:moments}; by sample values we refer to the mean, skewness and kurtosis values of the payoff distribution which are constant and plotted with dashed lines.
Similarly, doubling the number of trials brings in no further improvement in learning the distributional moments.
We note that skewness and mean of the payoff distributions are learned and generalized quite well on average.
Another reason we present Figure \ref{plot:moments} is to visually show that the inability to learn the kurtosis tends to reflect on the Wasserstein Distance, as it is a major component of the value distribution we aim to learn; we further note that Figure \ref{plot:200_epocs} is specifically chosen among unsuccessful trials to present a possible weakness and a cautionary suggestion.

\begin{figure}[h]
	\begin{subfigure}{0.5\textwidth}
		\centering
		\includegraphics[ width=0.8\linewidth]{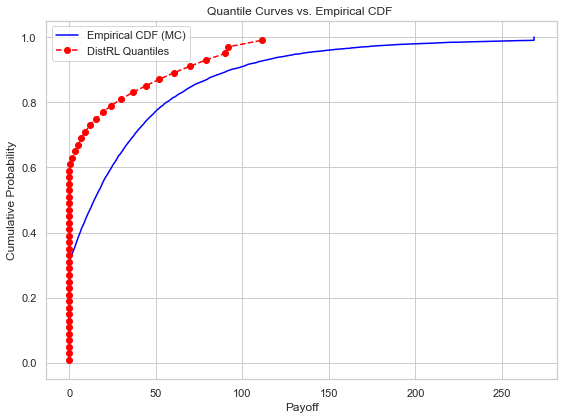}
		\caption{Only DistRL Clipped}
		\label{fig:d}
	\end{subfigure}
	\begin{subfigure}{0.5\textwidth}
		\centering
		\includegraphics[width=0.8\linewidth]{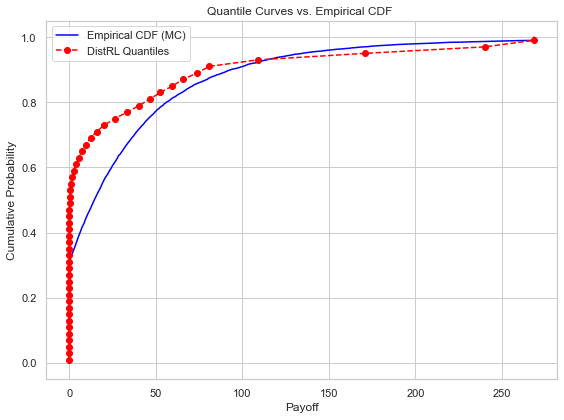}
		\caption{All Clipped}
		\label{2:d}
	\end{subfigure}
	\caption{Effects of Clipped Payoffs in a Larger Sample}
	\label{plot:clipped}
\end{figure}
Another point on clipping is the necessity of clipping on both training and test (evaluation) data.
An interesting example is given in Figure \ref{plot:clipped} in which an implausible problem of learning on a high spectrum of payoffs is designed with the initial price 120. 
In  both cased Monte Carlo generated payoffs are clipped at 200, but only in Figure \ref{2:d}; both are clipped at 200, while in Figure \ref{fig:d} the learning algorithm is designed to learn a considerably narrowed distribution, so we observe a horrifying amount of underestimation. 
We remark that while clipping payoffs to generate a better learning experience might be a good idea, an extreme caution should be applied.
Especially in cases of a diverse portfolios of derivatives; this issue is likely to amplify.

\begin{figure}[h]
\begin{subfigure}{0.5\textwidth}
	\centering
	\includegraphics[ width=0.8\linewidth]{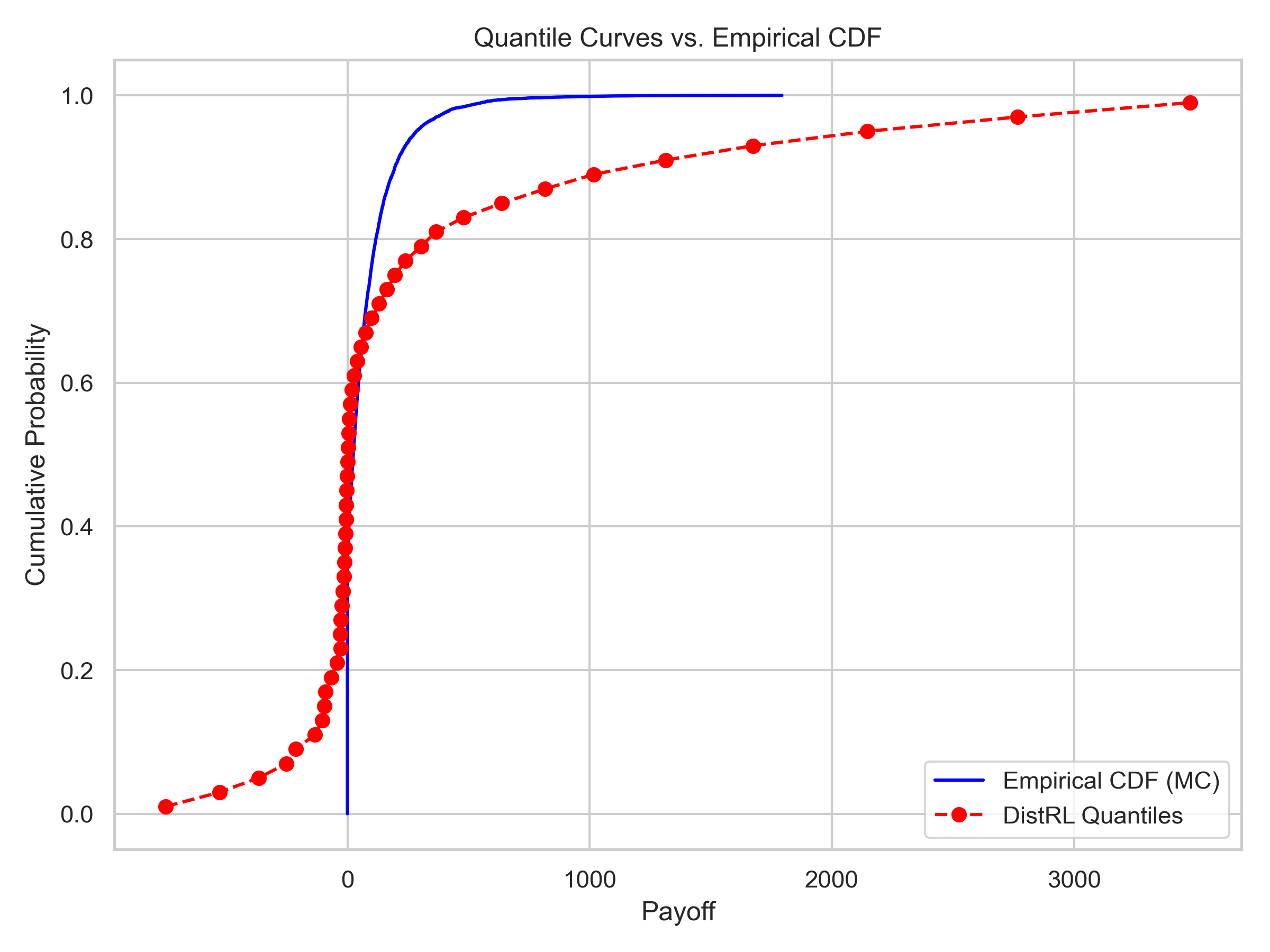}
	\caption{Example of negative prices}
	\label{plot:a}
\end{subfigure}
\begin{subfigure}{0.5\textwidth}
	\centering
	\includegraphics[width=0.8\linewidth]{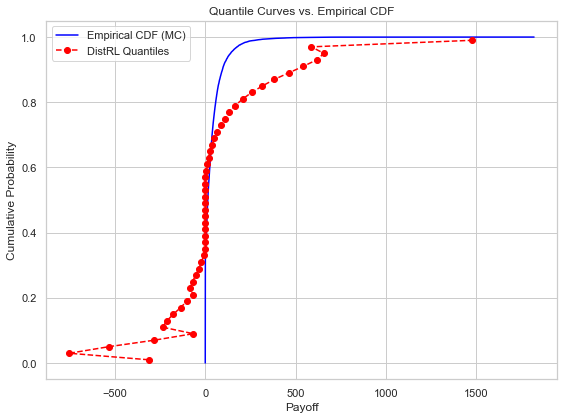}
	\caption{Example of negative prices}
	\label{plot:b}
\end{subfigure}
\caption{Effects of Non-clipped Gradients in a Larger Sample}
\label{plot:negative_prices}
\end{figure}

We finally present, perhaps the most critical necessity in our framework, the effects of non-clipped gradients.
In Figure \ref{plot:negative_prices}, we present specifically selected results out of experiments in which we enforced the algorithm to learn on payoffs worth from 0 to approximately 2000 and to over 3000 generated again from an initial price of 120 and strike price of 100.
While we are aware the absurdity of the designed problem, we remark that stretching such boundaries might constitute an interesting example.
In these cases, gradients  are not clipped hence leading to explosions in updates.
This is also amplified by the sensitivity in the tail of quantile regression.   
Although on a realistic set of data to train, such negative prices generated by the framework is highly unlikely (at least not observed during such cases) and the underlying method in our framework could lead to such extremes if the learning environment is not carefully designed.

We provide several other important points to consider such as sparse learning signals or strictly sharp gradients near maturity, which might be caused by sudden spikes in the underlying price process.
To address this concern; reward shaping could have been an important approach.
In a case of ours, reward shaping is not of help as we are constrained by the financial structure of path dependence.
Another point is the quantile spacing induced by the modeler.
In our manuscript,we find 50 quantiles sufficient, yet with an unrealistic selection one runs the likelihood of learning several quantiles accurately but missing the overall shape of the value distribution.
Though we state the RBF is so far capable of adapting to rapid or sudden distributional shifts in value targets, this might not always be the case as the payoff distribution could change drastically over time.
Although our design of state appeals reasonable in our case, any mismatch might lead to unstable training or even garbage quantiles causing absurd updates in the core of the framework.
Such issues should always be under consideration in applying the framework we presented in this manuscript, and special attention must be given to the underlying distribution to learn from.
Although it is far from being perfect, on the cases that our framework is flexible (any distribution and any payoff), scalable, and interpretable  and we argue that it could be a powerful tool.
Besides, our framework is online-compatible as the learning continuous as new data (or episodes) arrives.




\section{Conclusion}	

Model-free approaches, generally, are far from perfect, yet they have undeniably appealing properties, \cite{moody2001learning, nevmyvaka2006reinforcement}.
A major one of these is the unnecessity of a strict distributional assumption.
A special case to our framework is the full distributional access which could lead to easily implemented risk controls.
Even though we primarily focused on the pricing part of options; we strongly believe that Distributional Reinforcement Learning could be more and more appealing to risk management principles as access to quantiles and tail probabilities are possible without strong assumptions, and once trained, a thorough generalization could be possible effectively.

Our preference in path-dependent options are strictly related to the promises of DistRL.
Any path-dependent option could be priced via the framework we present in our study, even American options if extend it with policy optimization. 
We leave this important future perspective as a further study.
 
The framework, in a way, learns via episodic simulations (i.e. epochs).
This could become especially powerful in data-limited setting, as each epoch could be randomized to some extent to keep training the algorithm.
Finally, it could be modular with function approximation  meaning that the RBF approximation could be replaced with neural networks for especially high-dimensional approximation.
A natural alternative is to directly fit a parametric (or non-parametric) distribution to Monte Carlo payoff samples at each state. 
While feasible in principle, this approach has several critical limitations compared to Distributional Reinforcement Learning.
Fitting a distribution at each time-step or state ignores temporal structure and directly fitting distributions at every point $s_t \in \mathbb{R}^d$ requires a regression of distributions over a continuous state space but DistRL learns quantiles as continuous functions of state.
Hence we achieve by reformulating option pricing as a recursive estimation of $Z_t$, we recover a non-parametric, model-free alternative to traditional  methods.

We, throughout the study followed the principle of designing no further than needed.
As the learning problem is a data or scenario based problem, we make sure no extra steps are taken to present results more positively than they are.
This is the reason we do not employ possible techniques to increase the capability of the learning algorithm but rather allow it go on its own and on its own alone.
We finally remark the necessity of testing our framework on stochastic volatility and interest rate scenarios along with real life option payoffs.

\backmatter


\section*{Compliance with Ethical Standards}


\bmhead*{Funding} No funding received in the making of this manuscript.

\bmhead*{Competing Interests}
 There are no financial or non-financial interests   directly or indirectly related to the work submitted for publication.





\bigskip

\begin{appendices}






\end{appendices}


\bibliography{dist_rl_path_dep}
\end{document}